\documentclass[aps,pra,twocolumn,showpacs,showkeys,superscriptaddress,nobalancelastpage,nofootinbib,floatfix,longbibliography]{revtex4-2}

\usepackage{amsmath,amssymb,amsthm}
\usepackage{graphicx}
\usepackage{physics}
\usepackage{enumitem}
\usepackage{xspace}

\usepackage{orcidlink}

\usepackage{url}
\usepackage{hyperref}
\usepackage{color}
\definecolor{refcolor}{RGB}{160,35,0}
\definecolor{hrefcolor}{RGB}{0,35,190}
\hypersetup{
    colorlinks,
    citecolor=refcolor,
    filecolor=refcolor,
    linkcolor=hrefcolor,
    urlcolor=hrefcolor
}

\def\({\left(}
\def\){\right)}
\def\[{\left[}
\def\]{\right]}

\newcommand{\hilbert}{\mathcal{H}}
\newcommand{\pobs}[1]{#1}
\newcommand{\obs}[1]{\mathsf{\pobs{#1}}}
\newcommand{\oper}[1]{\widehat{\obs{#1}}}

\newcommand{\x}{\mathbf{x}}

\newcommand{\p}{\boldsymbol{p}}

\newcommand{\eg}{\textit{eg.}\xspace}
\newcommand{\etc}{\textit{etc}\xspace}
\newcommand{\etal}{\textit{et al.}\xspace}

\newtheorem{theorem}{Theorem}
\newtheorem{test}{Test}
\newtheorem{question}{Question}
\newenvironment{answer}{\begin{proof}[Answer]}{\end{proof}}

\theoremstyle{remark}

\newtheoremstyle{nospace} % name of the style
  {3pt}                   % Space above
  {3pt}                   % Space below
  {\itshape}              % Body font
  {}                      % Indent amount
  {\bfseries}             % Theorem head font
  {.}                     % Punctuation after theorem head
  {.5em}                  % Space after theorem head
  {\thmname{#1}\thmnumber{#2}\thmnote{ (#3)}} % <--- HEADER SPECIFICATION

\theoremstyle{nospace}
\newtheorem{QT}{QT}
\newcommand{\qtref}[1]{{\normalfont{QT\ref{#1}}}}

\newcommand{\HSF}{{\normalfont{HSF}}\xspace}%{\ref{def:HSF}\xspace}
\newcommand{\HSFiso}{{\normalfont{HSF$'$}}\xspace}%{\ref{def:HSF}\xspace}

\newtheorem{defCustom}{}
\makeatletter
\newcommand{\setdefCustomtag}[1]{% \setdefCustomtag{<tag>}
  \let\oldthedefCustom\thedefCustom% Store \thedefCustom
  \renewcommand{\thedefCustom}{{\normalfont\textbf{#1}}}% Redefine it to a fixed value
  \g@addto@macro\enddefCustom{% At \end{defCustom}, ...
    \global\let\thedefCustom\oldthedefCustom}% ...restore \thetheorem
  }
\makeatother

\allowdisplaybreaks

\begin{document}

\title{No change in Hilbert space fundamentalism}
\author{\orcidlink{0000-0002-2765-1562}\ Cristi Stoica}
%\email[]{cristi.stoica@theory.nipne.ro,holotronix@gmail.com}
\affiliation{Dept. Th. Physics, NIPNE-HH, Bucharest, Romania. \href{mailto:cristi.stoica@theory.nipne.ro}{cristi.stoica@theory.nipne.ro},  \href{mailto:holotronix@gmail.com}{holotronix@gmail.com}}

\begin{abstract}
Hilbert space fundamentalism (HSF) states that everything about the physical world is encoded in the Hamiltonian operator and the state vector (as a unit vector, not a wavefunction, which requires additional specification of a configuration space, a position basis, or the position observables). That all structures needed to describe reality, including subsystems, space, fields, emerge from these.

I show that HSF can't account for our observations that the physical world changes in time.
\end{abstract}

\maketitle

The quantum theory describing our physical world can be specified by the following:

\begin{QT}
\label{QT:HSF}
A Hilbert space $\hilbert$, a Hermitian operator $\oper{H}$ on $\hilbert$, and a unit vector $\ket{\psi(t)}\in\hilbert$ representing the total state. The time evolution of $\ket{\psi(t)}$ is generated by $\oper{H}$:
\begin{equation}
\label{eq:time-evol}
\ket{\psi(t)}=e^{-\frac{i}{\hbar}\oper{H}t}\ket{\psi(0)}.
\end{equation}
\end{QT}

\begin{QT}
\label{QT:OBS}
An association of physically observable properties to Hermitian operators on $\hilbert$.
These are of the form

\begin{tabular}{p{5cm} l l}
&&\\
  position~of~particle~$j$ & $\mapsto$ & $(\widehat{x}_j,\widehat{y}_j,\widehat{z}_j)$ \\
  momentum~of~particle~$j$ & $\mapsto$ & $(\widehat{p}_{x_j},\widehat{p}_{y_j}, \widehat{p}_{z_j})$ \\
  $\ldots$ &  & $\ldots$ \\
\end{tabular}

or similarly from quantum field operators.
\end{QT}

Knowing which properties correspond to each subsystem (\eg particle $j$) is sufficient to determine a decomposition of $\hilbert$ as a tensor product of subsystem spaces~\cite{ZanardiLidarLloyd2004QuantumTensorProductStructuresAreObservableInduced}.
\begin{equation}
\label{eq:hilbert-tensor}
\hilbert \cong \hilbert_1\otimes\hilbert_2\otimes\ldots,
\end{equation}

\begin{QT}
\label{QT:PROJ}
When measuring an observable $\oper{A}=\sum_a a \dyad{a}$ we obtain an eigenvalue $a$ of $\oper{A}$ with probability $\abs{\braket{a}{\psi}}^2$; (our description of) the system updates from $\ket{\psi}$ to $\ket{a}$.
\end{QT}

From what we know, \qtref{QT:HSF}--\ref{QT:PROJ} provide a complete description of the facts about our world that can be accessed empirically.
But can we get the same for less?
For example, Everett~\cite{Everett1957RelativeStateFormulationOfQuantumMechanics} proposed that \qtref{QT:PROJ} is explained by \qtref{QT:HSF} and \qtref{QT:OBS}. But it seems that both \qtref{QT:HSF} and \qtref{QT:OBS} are still needed~\cite{Vaidman2016AllIsPsi,Stoica2022SpaceThePreferredBasisCannotUniquelyEmergeFromTheQuantumStructure}, since not the bare state vector $\ket{\psi(t)}$ describes the world, but the wavefunction
\begin{equation}
\label{eq:wavefunction}
\psi(\x,t)=\braket{\x}{\psi(t)}
\end{equation}
where $\ket{\x}=\ket{\ldots,\widehat{x}_j,\widehat{y}_j,\widehat{z}_j,\ldots}$ is a common eigenvector of the position observables.
Without observables, $\ket{\psi}$ is just a unit vector like any other unit vector in $\hilbert$. And the Hamiltonian $\oper{H}$ is like any other operator $\oper{U}\oper{H}\oper{U}^\dagger$ obtained by a unitary transformation $\oper{U}$.

But still, can't we just derive \qtref{QT:OBS} from \qtref{QT:HSF}?
After all, when we get the Hilbert space, we get for free all operators.
What do we gain by naming them ``position'' or ``momentum''? Don't we get all we need just from their spectra and their relations with other operators?
The map from \qtref{QT:OBS} is useful for thinking, but does it play a fundamental role in a purely structural theory?

It seems we can access empirically only ``clicks'' in the detectors or in our sense organs, only quantum information, from which we can infer correlations and reconstruct the relations and structures of our world. All we can know and describe using math are relations, structure. But maybe that's all there is to know.
This motivates a new paradigm shift, the suggested realization that \qtref{QT:HSF} is all we need to describe the world~\cite{CarrollSingh2019MadDogEverettianism,Carroll2021RealityAsAVectorInHilbertSpace}.
\setdefCustomtag{HSF}
\begin{defCustom}[Hilbert space fundamentalism]
\label{def:HSF}
The triple $(\hilbert,\oper{H},\ket{\psi})$ encodes a complete unambiguous description of the physical world. \normalfont{Equivalently,}
\end{defCustom}

\setdefCustomtag{HSF$'$}
\begin{defCustom}[Equivalent formulation of \HSF]
\label{def:HSF-iso}
Two triples $(\hilbert,\oper{H},\ket{\psi})$ and $(\hilbert',\oper{H'},\ket{\psi'})$ describe one and the same physical reality if and only if they are isomorphic, in the sense that there is a unitary map $\oper{U}:\hilbert\to\hilbert'$ so that
\begin{equation}
\label{eq:HSF-iso}
\oper{H'}=\oper{U}\oper{H}\oper{U}^{-1} \text{ {\normalfont and} } \ket{\psi'}=\oper{U}\ket{\psi}.
\end{equation}
\end{defCustom}

If \HSF is correct, \qtref{QT:HSF} should be able to explain the world just as well as \qtref{QT:HSF}$+$\qtref{QT:OBS}. In particular, it should be able to describe how things change. So a basic test is:
\begin{test}
\label{test:HSF}
Can \HSF describe unambiguously how the world changes in time?
\end{test}

It may be more difficult for \qtref{QT:HSF} to do this without \qtref{QT:OBS}, but \emph{if} we can decode the observable map from \qtref{QT:OBS} from $(\hilbert,\oper{H},\ket{\psi})$, everything should be equally easy.
All we need is to explain how subsystems emerge~\cite{CotlerEtAl2019LocalityFromSpectrum}, how space emerges~\cite{CarrollSingh2019MadDogEverettianism,Carroll2021RealityAsAVectorInHilbertSpace}, how pointer states emerge~\cite{Zurek1998DecoherenceEinselectionAndTheExistentialInterpretation,LombardiEtAl2012TheProblemOfIdentifyingSystemEnvironmentDecoherence} and so on.
And if this reconstruction doesn't give us back standard quantum theory, but a new quantum gravity theory~\cite{CarrollSingh2019MadDogEverettianism,Carroll2021RealityAsAVectorInHilbertSpace}, even better.
But does \HSF pass Test~\ref{test:HSF}?

\begin{theorem}
\label{thm:HSF-not-QT}
\HSF can't describe unambiguously how the world changes in time.
\end{theorem}
\begin{proof}
Taking as the unitary map $\oper{U}$ from~\eqref{eq:HSF-iso} the unitary operator $\oper{U}_t:=e^{-\frac{i}{\hbar}\oper{H}t}$, and using $\oper{U}_t\oper{H}=\oper{H}\oper{U}_t$, we obtain
\begin{equation}
\label{eq:HSF-auto-time}
\oper{H}=\oper{U}_t\oper{H}\oper{U}_t^{\dagger} \text{ and } \ket{\psi(t)}=\oper{U}_t\ket{\psi(0)}.
\end{equation}

Then, according to \HSFiso, the triples $(\hilbert,\oper{H},\ket{\psi(t)})$ and $(\hilbert,\oper{H},\ket{\psi(0)})$ describe one and the same physical reality.
On the other hand, the triples $(\hilbert,\oper{H},\ket{\psi(t)})$ and $(\hilbert,\oper{H},\ket{\psi(0)})$ are supposed to describe different physical realities, since the world change in time.
Therefore, \HSF can't describe changes in the world.
\end{proof}

The rest of this article addresses potential objections.

%------------------------------------------------------------%
\subsection*{Questions }

\begin{question}
\label{q:mwi}
Isn't \HSF only about Everett's theory?
\end{question}
\begin{answer}
\HSF was advocated in the context of Everett's theory in~\cite{CarrollSingh2019MadDogEverettianism,Carroll2021RealityAsAVectorInHilbertSpace}, but it has a broader scope. Theorem~\ref{thm:HSF-not-QT} is about whether \qtref{QT:OBS} is reducible to \qtref{QT:HSF}, independently of whether \qtref{QT:PROJ} is also reducible to them.
\end{answer}

\begin{question}
\label{q:evol-not-auto}
The operator $\oper{U}_t$ gives the system's time evolution, so time changes. Is it then allowed to use $\oper{U}_t$ as an automorphism of $(\hilbert,\oper{H},\ket{\psi(0)})$ at a fixed time $t=0$?
\end{question}
\begin{answer}
$\oper{U}_t$ gives the time evolution, but this doesn't exclude it from the automorphism group of $(\hilbert,\oper{H},\ket{\psi(0)})$ at $t=0$. As an automorphism, it shows that \HSF doesn't distinguish the world's states at different times, because ``isomorphic triples describe one and the same reality''.
\end{answer}

\begin{question}
\label{q:wavefunction}
Given that $\psi(0)$ and $\psi(t)$ are different wavefunctions, how can you say that $(\hilbert,\oper{H},\ket{\psi(0)})$ and $(\hilbert,\oper{H},\ket{\psi(t)})$ describe one and the same physical reality?
\end{question}
\begin{answer}
I don't say this, \HSFiso says it, thus refuting itself.

In \qtref{QT:HSF}$+$\qtref{QT:OBS} we can talk about the wavefunction $\psi(\x,t)$. But once we embrace \HSF, we get rid of \qtref{QT:OBS}, including the position observables defining the wavefunction $\psi(\x,t)=\braket{\x}{\psi(t)}$ or the position configuration space on which it propagates.
\HSF can't presuppose them, it has to derive them unambiguously from \qtref{QT:HSF}.
\end{answer}

\begin{question}
\label{q:unit-vector}
But aren't $\ket{\psi(t)}$ and $\ket{\psi(0)}$ different? Then why do you say that in \HSF $(\hilbert,\oper{H},\ket{\psi(0)})$ and $(\hilbert,\oper{H},\ket{\psi(t)})$ describe one and the same physical reality?
\end{question}
\begin{answer}
I don't say this, \HSFiso says it, thus refuting itself.

The vectors $\ket{\psi(0)}$ and $\ket{\psi(t)}$ can be distinct, but $(\hilbert,\oper{H},\ket{\psi(0)})$ and $(\hilbert,\oper{H},\ket{\psi(t)})$ are isomorphic, so \HSFiso says that they describe one and the same physical reality. So \HSF misses the world's changes in time.
\end{answer}

\begin{question}
\label{q:observables}
Can't we use observables to extract the full description of the world from $\ket{\psi}$?
Wouldn't this allow \HSF to distinguish the world's states at different times?
\end{question}
\begin{answer}
We can talk about using observables to extract facts about reality from $\ket{\psi}$ in \qtref{QT:HSF}$+$\qtref{QT:OBS}, but how can we do this in \HSF? When \HSF got rid of \qtref{QT:OBS}, it lost the map between observables and the physical properties they represent.
So in \HSF you can use operators to extract data from $\ket{\psi}$, but you can't know what physical properties in the world those data represent.
There are infinitely many unitary operators $\oper{U}$ that commute with $\oper{H}$, and so there are infinitely many position-like operators $\widehat{\x}':=\oper{U}\widehat{\x}\oper{U}^\dagger$, all with the same spectra as $\widehat{\x}$ and in the same relation with $\oper{H}$. But in general they give different answers if you use them to compute $\braket{\x'}{\psi}$. So, again, no unique wavefunction emerges from \HSF, and there are no unambiguous answers to pertinent questions like ``what result did the pointer show?'', or ``where is Alice?''.
\end{answer}

\begin{question}
\label{q:obs-fixed}
In \qtref{QT:HSF}$+$\qtref{QT:OBS}, observables like $\widehat{\x}$ or $\widehat{\p}$ and the decomposition into subsystems~\eqref{eq:hilbert-tensor} don't change in time (in the Schr\"odinger picture used here).
Didn't the proof of Theorem~\ref{thm:HSF-not-QT} assume that they have to change too, covariantly with $\ket{\psi(t)}$, and this is why you mistakenly think that \HSF can't describe changes in the world?
\end{question}
\begin{answer}
The whole point of \HSF is that it got rid of the physical meanings of observables and the decomposition into subsystems, to the effect that any triple unitarily equivalent with $(\hilbert,\oper{H},\ket{\psi(0)})$ describes one and the same state of the world. Then $(\hilbert,\oper{H},\ket{\psi(t)})$ describes one and the same state of the world as $(\hilbert,\oper{H},\ket{\psi(0)})$.

I know it feels that you can find such observables and fix them, and then everything works as in \qtref{QT:HSF}$+$\qtref{QT:OBS}.
But if you do this at $t=0$ and repeat the process to recover them at another time $t\neq 0$, if you want to avoid Theorem~\ref{thm:HSF-not-QT}, you'll either have to change the rules to get them ``unevolved'' at every time $t$, or to use a rule that gives multiple solutions, and pick the one that fits and use it at all times.
But then, you'll have to make additional choices not specified in \qtref{QT:HSF}, like we do in \qtref{QT:HSF}$+$\qtref{QT:OBS}.
\end{answer}

\begin{question}
\label{q:symmetry-breaking}
But the whole point of \HSF is that it got rid of the physical meanings of observables and the decomposition into subsystems, precisely because they can be uniquely reconstructed from the mathematical relations only.
So why wouldn't this work?
\end{question}
\begin{answer}
If $(\hilbert,\oper{H},\ket{\psi(0)})$ and $(\hilbert,\oper{H},\ket{\psi(t)})$ describe one and the same world, no construction done in both of them can give different empirical results for different times.
\end{answer}

\begin{question}
\label{q:emerge}
Did you consider that all structures and observables needed to complete the description given by $(\hilbert,\oper{H},\ket{\psi})$ can emerge from $(\hilbert,\oper{H},\ket{\psi})$? See~\cite{CotlerEtAl2019LocalityFromSpectrum,CarrollSingh2019MadDogEverettianism,Carroll2021RealityAsAVectorInHilbertSpace}.
\end{question}
\begin{answer}
I did, this was the starting point. In \cite{Stoica2022SpaceThePreferredBasisCannotUniquelyEmergeFromTheQuantumStructure} I gave a fully general proof that, if its construction is fully invariant as it should be in \HSF, such a preferred structure 
\begin{enumerate}
	\item 
either is ambiguously defined: you obtain many isomorphic structures, giving incompatible answers to questions about the world that in \qtref{QT:HSF}$+$\qtref{QT:OBS} have clear definite answers at each time,
	\item 
or emerges uniquely up to physically observable differences,  and then it's unable to distinguish the state of the world at different times.
\end{enumerate}

Since the proof from \cite{Stoica2022SpaceThePreferredBasisCannotUniquelyEmergeFromTheQuantumStructure} is quite long and complicated, let me answer this in the simpler setup from this paper.

Suppose that some preferred structure $\mathcal{S}$ that quantum theory needs emerges uniquely. It can be the tensor product structure (TPS), the position observables, the 3D space, the spacetime, a preferred pointer basis \etc. 
We care about these structures because we need them to extract data from $\ket{\psi}$. So let's consider a question $\mathcal{Q}$ about the world, whose answer depends on the structure $\mathcal{S}$ and $\ket{\psi}$. 
For example, entanglement depends on the TPS.
According to \HSFiso, the question should have the same answer in any two isomorphic triples $(\hilbert,\oper{H},\ket{\psi})$ and $(\hilbert',\oper{H'},\ket{\psi'})$.
In particular, it should have the same answer in $(\hilbert,\oper{H},\ket{\psi(0)})$ and $(\hilbert,\oper{H},\ket{\psi(t)})$, for any $t$.
But then, our emergent structure $\mathcal{S}$ can't help us get answers that are different at different times.
For example, the TPS has to be able to define unambiguously the entanglement entropies of subsystems, so it has to be unique. But if a unique TPS emerges only from $(\hilbert,\oper{H},\ket{\psi(t)})$ independently of $t$, as \HSFiso implies, the Hamiltonian must lack interactions and it can't change entanglement~\cite{Stoica2024DoesTheHamiltonianDetermineTheTPSAndThe3dSpace}. But in the real world it can.
In general, \HSF, supplemented with any structures that may emerge uniquely from $(\hilbert,\oper{H},\ket{\psi})$, can't describe changing things.

This is simply because, since $(\hilbert,\oper{H},\ket{\psi})$ gives an incomplete description of the world, it can't bootstrap itself into completion.
If it can't describe how things change, any construction you make without breaking its symmetry by hand also can't describe how things change.
\end{answer}

\begin{question}
\label{q:connection-past-work}
Is there a connection between the result from~\cite{Stoica2022SpaceThePreferredBasisCannotUniquelyEmergeFromTheQuantumStructure} which you just summarized in the answer to Question~\ref{q:emerge}, and Theorem~\ref{thm:HSF-not-QT} from this paper?
\end{question}
\begin{answer}
In~\cite{Stoica2022SpaceThePreferredBasisCannotUniquelyEmergeFromTheQuantumStructure} I show that any preferred structure hoped to emerge only from $(\hilbert,\oper{H},\ket{\psi})$ either can't exhibit physically observable differences between states at different times, or it is not physically unique, so it's ambiguous.
When the solution is not unique, there are solutions that answer questions about the world as it is at different times, but we don't know which is which, so they too can't be used to describe change. For example, a TPS as obtained in~\cite{CotlerEtAl2019LocalityFromSpectrum} is unique up to symmetries of the Hamiltonian, so it can't define the entanglement entropy unambiguously, since different TPSs give different answers, valid at different other times~\cite{Stoica2024DoesTheHamiltonianDetermineTheTPSAndThe3dSpace}. The same applies to the emergence of a pointer basis or of 3D space~\cite{CarrollSingh2019MadDogEverettianism,Carroll2021RealityAsAVectorInHilbertSpace}. To summarize,~\cite{Stoica2022SpaceThePreferredBasisCannotUniquelyEmergeFromTheQuantumStructure} proves \HSF's inability to detect changes, just like the present Theorem~\ref{thm:HSF-not-QT}, but less directly.
\end{answer}

\begin{question}
\label{q:trivial}
If \HSF so trivially refutes itself, why did you bother to write and publish a 68 page long paper~\cite{Stoica2022SpaceThePreferredBasisCannotUniquelyEmergeFromTheQuantumStructure} and sequels~\cite{Stoica2023PrinceAndPauperQuantumParadoxHilbertSpaceFundamentalism,Stoica2024DoesTheHamiltonianDetermineTheTPSAndThe3dSpace}, building an entire framework for \HSF, only to prove weaker and less evident results than the one in the present paper?
\end{question}
\begin{answer}
Maybe I was too caught up in the game myself.
\end{answer}

\begin{question}
\label{q:structuralism-physicalism}
Along with the Hilbert space $\hilbert$, we get for free all operators on $\hilbert$, so all observables are already in \qtref{QT:HSF}.
The map from \qtref{QT:OBS} only labels them, names them as positions, momenta \etc.
How can this naming, which is just a convention, turn the ambiguous unchanging $(\hilbert,\oper{H},\ket{\psi})$ into a changing world as in \qtref{QT:HSF}$+$\qtref{QT:OBS}?
\end{question}
\begin{answer}
This is a great question, and it is where things become interesting~\cite{Stoica2023AreObserversReducibleToStructures,Stoica2025StructuralRealismAndTheUnderdeterminationOfThePhysicalMeaning,Stoica2024ObservationAsPhysication}.
\end{answer}

\begin{question}
\label{q:survives}
Does Theorem~\ref{thm:HSF-not-QT} leave any hope for \HSF?
\end{question}
\begin{answer}
Theorem~\ref{thm:HSF-not-QT} says that \HSF can't describe things that can change, but it doesn't say that it can't describe things that don't change, facts about the form of the physical law.
Here there is room for progress.
For example, Cotler \etal~\cite{CotlerEtAl2019LocalityFromSpectrum} studied TPSs in which the Hamiltonian is local, in the sense that it contains only interactions between a small number of subsystems.
Their result was widely misunderstood as showing the emergence of a unique TPS from the Hamiltonian's spectrum alone, but it is in fact about the form of the Hamiltonian.
I suspect other similarly interesting results can be discovered.

So the part of \HSF dealing only with the physical law, which doesn't change, evades Theorem~\ref{thm:HSF-not-QT}.
\end{answer}

%------------------------------------------------------------%


\begin{thebibliography}{10}

\bibitem{ZanardiLidarLloyd2004QuantumTensorProductStructuresAreObservableInduced}
P.~Zanardi, D.A. Lidar, and S.~Lloyd.
\newblock Quantum tensor product structures are observable induced.
\newblock {\em Phys. Rev. Lett.}, 92(6):060402, 2004.

\bibitem{Everett1957RelativeStateFormulationOfQuantumMechanics}
H.~Everett.
\newblock ``{R}elative state'' formulation of quantum mechanics.
\newblock {\em Rev. Mod. Phys.}, 29(3):454--462, Jul 1957.

\bibitem{Vaidman2016AllIsPsi}
L.~Vaidman.
\newblock All is $\psi$.
\newblock {\em J. Phys. Conf. Ser.}, 701:012020, 2016.

\bibitem{Stoica2022SpaceThePreferredBasisCannotUniquelyEmergeFromTheQuantumStructure}
O.C. Stoica.
\newblock 3d-space and the preferred basis cannot uniquely emerge from the
  quantum structure.
\newblock {\em Adv. Theor. Math. Phys.}, 26(10):3895---3962, 2022.
\newblock \href{http://arxiv.org/abs/2102.08620}{arXiv:2102.08620}.

\bibitem{CarrollSingh2019MadDogEverettianism}
{Carroll \& Singh}.
\newblock Mad-dog {E}verettianism: {Q}uantum {M}echanics at its most minimal.
\newblock In A.~Aguirre, B.~Foster, and Z.~Merali, editors, {\em What is
  Fundamental?}, pages 95--104. Springer, 2019.

\bibitem{Carroll2021RealityAsAVectorInHilbertSpace}
S.M. Carroll.
\newblock Reality as a vector in {H}ilbert space.
\newblock In Valia Allori, editor, {\em Quantum mechanics and fundamentality},
  volume 460, pages 211--224. Springer Nature, 2022.

\bibitem{CotlerEtAl2019LocalityFromSpectrum}
J.S. Cotler, G.R. Penington, and D.H. Ranard.
\newblock Locality from the spectrum.
\newblock {\em Comm. Math. Phys.}, 368(3):1267--1296, 2019.

\bibitem{Zurek1998DecoherenceEinselectionAndTheExistentialInterpretation}
W.H. Zurek.
\newblock {Decoherence, Einselection, and the Existential Interpretation (The
  Rough Guide)}.
\newblock {\em Phil.\ Trans.\ Roy.\ Soc.\ London Ser. A}, A356:1793, 1998.
\newblock \href{http://arxiv.org/abs/quant-ph/9805065}{arXiv:quant-ph/9805065}.

\bibitem{LombardiEtAl2012TheProblemOfIdentifyingSystemEnvironmentDecoherence}
Olimpia Lombardi, S.~Fortin, and M.~Castagnino.
\newblock The problem of identifying the system and the environment in the
  phenomenon of decoherence.
\newblock In {\em EPSA Philosophy of Science: Amsterdam 2009}, pages 161--174.
  Springer, 2012.

\bibitem{Stoica2024DoesTheHamiltonianDetermineTheTPSAndThe3dSpace}
O.C. Stoica.
\newblock Does the {H}amiltonian determine the tensor product structure and the
  3d space?
\newblock 2024.
\newblock \href{http://arxiv.org/abs/2401.01793}{arXiv:2401.01793}.

\bibitem{Stoica2023PrinceAndPauperQuantumParadoxHilbertSpaceFundamentalism}
O.C. Stoica.
\newblock The prince and the pauper: {A} quantum paradox of {H}ilbert-space
  fundamentalism.
\newblock 2023.
\newblock \href{http://arxiv.org/abs/2310.15090}{arXiv:2310.15090}.

\bibitem{Stoica2023AreObserversReducibleToStructures}
O.C. Stoica.
\newblock Are observers reducible to structures?
\newblock 2023.
\newblock \href{https://arxiv.org/abs/2307.06783}{arXiv:2307.06783}.

\bibitem{Stoica2025StructuralRealismAndTheUnderdeterminationOfThePhysicalMeaning}
O.C. Stoica.
\newblock Structural realism and the underdetermination of the physical
  meaning.
\newblock {\em J. Phys.: Conf. Ser.}, 2948(1):012014, 2025.

\bibitem{Stoica2024ObservationAsPhysication}
O.C. Stoica.
\newblock Observation as physication. {A} single-world unitary no-conspiracy
  interpretation of quantum mechanics.
\newblock 2024.
\newblock \href{https://arxiv.org/abs/2412.09669}{arXiv:2412.09669}.

\end{thebibliography}
\end{document}